\documentclass[submission,copyright,creativecommons]{eptcs}
\providecommand{\event}{ACT 2020} 
\usepackage{breakurl}             
\usepackage{underscore}           
\usepackage{amsmath}
\usepackage{amssymb}
\usepackage{enumitem}
\usepackage{tikz}

\title{Products in a Category with Only One Object}
\author{Rick Statman
\institute{Carnegie Mellon University}
\email{rs31@webmail.math.cmu.edu}
}
\def\titlerunning{Products in a Category with Only One Object}
\def\authorrunning{Rick Statman}

\newcommand\hookto{\mathrel{\tikz{
   \path [use as bounding box] (-.04,-0.09) rectangle (.254,.01);
   \draw [>->] (-.02,0) to +(.266,0);
}}}
\newcommand\hooktto{\mathrel{\tikz{
   \path [use as bounding box] (-.04,-0.09) rectangle (.254,.01);
   \draw [>->>] (-.02,0) to +(.266,0);
}}}

\usepackage{amsthm}
\newtheorem{theorem}{Theorem}
\newtheorem{lemma}{Lemma}

\newtheorem*{corollary*}{Corollary}
\newtheorem{proposition}{Proposition}
\theoremstyle{definition}
\newtheorem{definition}{Definition}

\begin{document}

\maketitle

\begin{abstract}
We consider certain decision problems for the free model of the theory of Cartesian monoids. We introduce a model of computation based on the notion of a single stack one-way PDA due to Ginsburg, Greibach and Harrison. This model allows us to solve problems such as:
\begin{enumerate}[label=(\arabic*)]
\item Given a finite set B of elements and an element F, is F a product of members of B?
\item Is the submonoid generated by the finite set B infinite?
\end{enumerate}
for certain fragments of the free Cartesian monoid. These fragments include the submonoid of
right invertible elements and so our results apply to the Thompson-Higman groups.
\end{abstract}

\section{INTRODUCTION AND PRELIMINARIES}

The notion of a Cartesian monoid has been rediscovered many times. A list of those just known
to me would include
Jonsson-Tarski Algebras~\cite{jonsson61},
Vagabond Groups~\cite{thompson80},
Cantor algebras~\cite{smirnov71},
FP~\cite{backus78},
Cartesian Monoids~\cite{lambek80},
Freyd-Heller Monoids~\cite{freyd93},
TOPS~\cite{statman96},
Thompson-Higman Monoids~\cite{birget09},
and CP Monoids~\cite{gray16}.
The title of this paper has been pirated from Gray \& Pardue; the reason for this will become apparent presently.

The notion of Cartesian Monoid axiomatizes the idea of a monoid of functions on a set $S$ supporting a surjective pairing function $S \times S \to S$ lifted pointwise to $(S \to S) \times (S \to S) \to (S \to S)$. So a Cartesian Monoid $C = (M, *, I, L, R, \langle-\rangle)$
is a monoid $(M, *, I)$ together with elements $L, R:M$ and a map \mbox{$\langle - \rangle: M \times M \to M$}
satisfying:
\begin{align*}
L* \langle F, G \rangle &= F &&\text{(left projection)}
\\
R* \langle F, G \rangle &= G &&\text{(right projection)}
\\
\langle F, G \rangle * H &= \langle F*H, G*H \rangle && \text{(pointwise lifting)}
\\
\langle L, R \rangle &= I &&\text{(surjectivity)}
\end{align*}
More generally, one could consider pairing functions which might not
be surjective. In this case, the fourth condition would be omitted. The authors of the most recent
rediscovery call such structures Categorical Quasiproduct Monoids (CQP)~\cite[Section~2.4]{ginsburg67}. We are happy to
credit them with directing our attention to them.

Here it will be useful to review some properties of the free Cartesian Monoid CM. The four conditions have the equivalent rewrite system
\begin{enumerate}[label=(\arabic*)]
\item $L * \langle X, Y \rangle \hookto X$
\item $R * \langle X, Y \rangle \hookto Y$
\item $\langle X, Y \rangle * Z \hookto  \langle X * Z, Y * Z \rangle$
\item $I * X \hookto X$
\item $X * I \hookto X$
\item $\langle L*X, R*X \rangle \hookto X$
\item $\langle L, R \rangle\hookto I$
\end{enumerate}
modulo the associativity axioms. This rewrite system is equivalent to the conditions in the sense that the smallest associative congruence containing the rules as identities is the congruence generated by the conditions (axioms). We denote the monotone, reflexive, transitive closure of $\hookto$ by $\hooktto$. Here
monotone means replacing subexpressions by their rewrites.

\begin{theorem}[See \cite{statman96}]
\begin{enumerate}[label=(\arabic*)]
\item Every pair of equivalent expressions have a common rewrite.
\item Every sequence of rewrites eventually terminates in a unique normal form.
\item All of the last three rewrites can be done at the end. 
\end{enumerate}
\end{theorem}

\begin{corollary*}
The word problem for CM is solvable.
\end{corollary*}

The normal forms of Cartesian monoid expressions have a pleasing shape. They can be described as binary trees built up from $\langle - \rangle$ with strings of $L$'s and $R$'s at the leaves. These strings are built up by $*$ with $I$ as the empty string. We call these strings the \textit{shifts} of the normal form. Normal forms that use only the rewrites (1)--(5) have a similar shape.

Now in \cite{statman96} we proved the following.
\begin{theorem}
If $F$ and $G$ have distinct normal forms then there exist
$H$ and $K$ s.t. $H*F*K = L$ and $H*G*K = R$ in CM.
\end{theorem}

\noindent
This is the ``simplicity'' rediscovered in Birget~\cite{birget09}.

\section{CQ IS ALMOST SIMPLE}

Let CQ be the free Categorical Quasiproduct Monoid. Clearly
there exists a homomorphism from CQ onto the free Cartesian Monoid.
First, we add to Theorem~2 for the case of CQ.
\begin{lemma}
If $F_ {0} = F_ {1}$ in CM but the (1)--(5) normal forms of
$F_ {0}$ and $F_ {1}$ are distinct, say $U_ {0}, U_ {1}$
resp., then there exist H, K such that for $i= 0$ or $i=1$ we have:
\begin{align*}
H*U_{i}*K &\hooktto_{5,6,7} I
\\
H*U_{1-i} *K &\hooktto_{5,6,7} \langle L, R \rangle
\end{align*}
\end{lemma}
\begin{proof}
By induction on the sum of the lengths of $U_ {0}$ and $U_ {1}$.

\emph{Basis.} since the $U_{i}$ are
both in (1)--(5) normal form, they cannot both be shifts. Thus the shortest case has the form
$U_{i} = I$ and $U_{1-i} = \langle U, V \rangle$. Now $\langle U, V \rangle \hooktto_{5,6,7} I$. Thus $U \hooktto_{5,6,7} L$ and
$V \hooktto _{5,6,7} R$, and we are done, or $U \hooktto_{5,6,7} LX$ and $V \hooktto_{5,6,7} RX$ where $X$ must be a shift since $U_{1-i}$ is normal w.r.t. (1)--(5). But then $X = I$ and we are done again.

\emph{Induction Step.} In case that one of the $U_{i}$ is a shift is as in the basis case. Thus we can assume both $U_{i}$ begin with $\langle - \rangle$. Let $U_ {0}
= \langle V_ {0}, V_ {1} \rangle$ and $U_ {1} =  \langle W_ {0}, W_ {1} \rangle$. Then there exists a $j$ s.t. $V_{j} \not = W_{j}$. In case $j = 0$ we have $L*U_ {0} \hooktto (1) \hooktto V_ {0}$ and
$L*U_ {1} \hooktto (1) \hooktto W_ {0}$ and the induction hypothesis can be applied to $V_ {0}$ and $W_ {0}$. The case $j = 1$ is similar.
\end{proof}

\begin{proposition}
The homomorphism from CQ onto the free Cartesian Monoid is unique and the only
non-trivial homomorphism
of CQ.
\end{proposition}
\begin{proof}
If $F$ and $G$ are not equivalent in CM then there exist $H$ and $K$ as in Theorem~2. By Theorem~1 there exists a common rewrite which can be obtained by rewriting with all of the (5)--(7) rewrites at the end. Thus there exist U, V s.t.:
\begin{align*}
&H*F*H \hooktto_{1,2,3,4} U \hooktto_{5,6,7} L
\\
&H*G*H \hooktto_{1,2,3,4} V \hooktto_{5,6,7} R
\end{align*}
Now if $X*I$ occur in either $U$ or $V$ it can be eliminated by
$X*I \hookto X$, since $X$ cannot contain $\langle - \rangle$. Now, neither $U$ nor $V$ can contain $\langle X, I \rangle$, $\langle I, X \rangle$ or  $\langle L, R \rangle$ since this would prevent rewriting
to $L$ or $R$. This holds for any rewrite of $U$ or $V$ by (6). Thus only
the rewrite (6) is used. Hence there exist $X, Y$ such that $U = X*L$ and $V = Y*R$ (reverse rewrite (3) multiple times). But then $I = L* \langle I, I \rangle = U * \langle I, I \rangle = X$ and $I = R* \langle I, I \rangle = V* \langle I, I \rangle = Y$. So if $F$ and $G$ are identified by a homomorphism
of CQ then so are $X$ and $Y$. But then by Lemma~1 the homomorphism identifies $I$ and $\langle L, R \rangle$. Thus the homomorphism is a homomorphism of CM. This is impossible by \cite[Section 3]{statman96}.
\end{proof}

\section{G.G.H. PDA'S}

Now we would like to introduce a model for computing with CQ expressions which will allow us to prove many questions about CQ multiplication decidable. We cannot expect too much since it is obvious that the existential theory of CQ is undecidable~\cite[Section~9]{statman96}. For what follows
we shall refer to (1)--(5) normal form as \emph{CQ normal form}.

If $F$ is in CQ normal form, then the shift $S$ which when read from left to right describes the
position of the leaf of the tree of $F$ where the shift $S'$ resides, then the CQ normal form of $S*F$ is
$S'$, but for no initial segment $S''$ of $S$ is the normal form of $S''*F$ a shift. Our model is the G.G.H. notion~\cite{ginsburg67} of a one-way single stack pushdown automata. These are non-deterministic PDAs which can scan the current stack by a two way NDFA before reading the top most stack symbol,
executing a stack operation, changing state and reading the next input. The input alphabet
consists of CQ expressions $F$ in normal form taken from a given finite set, and input from left to
right.

When the initial contents of the stack are the string $S$, with top-to-bottom corresponding to right-to-left, and the input is $F_ {1}, F_{t}$, then we want the contents of the stack to be the normal form of $F_ {1}
* \cdots *F_{t}$. However, this may not be a shift. In this case we terminate the computation in failure.
The PDA operates as follows:
\begin{enumerate}[label=(\arabic*)]
\item It reads the input $F$.
\item It reads the stack to determine if it has the form $S' * S$ where the normal form of $S*F$ is a shift
$S''$. If not the computation fails.
\item If (2) succeeds it pops $S$ from the stack and pushes $S''$ onto the stack.
\end{enumerate}

Here we note that in (2) there are only finitely many $S$ to check so this can be implemented in a
deterministic G.G.H. Now G.G.H. proved that the sets accepted by G.G.H. PDAs are closed
under intersection and union but not complement. The deterministic one are closed under
complement but not intersection.

\begin{theorem}
Let $B$ be a finite set of CQ expressions. Then it is
decidable if $F$ is a product of members of~$B$.
\end{theorem}

\begin{proof}
We assume that all expression is in CQ normal form. First,
let $S_ {1}, \ldots, S_{n}$ be the complete list of shifts s.t. $S_{i}*F$
equals a shift $S_{i}$, but no initial segment of $S_{i}$ (right to left) does.

We construct deterministic G.G.H. machines $M_{i}, N_{i}$, for $i=1, \ldots, n$,
so that $M_{i}$ accepts all inputs $F_ {1}, \ldots, F_{k}$ s.t. $S_{i}*F_ {1} * \ldots *F_{k} = S_{i}$. For each
$i = 1, \ldots , n$ let $S''_{i}$ be $S_{i}$ with its last $L$ or $R$ (from right to left) deleted. Then $N_{i}$ is
constructed to accept all inputs $F_ {1}, \ldots, F_{k}$ s.t. the CQ normal form of $S''_{i}*F_ {1} *...*F_{k}$ is not a shift. Here we use closure of deterministic machines under complements.
Now by the work of G.G.H. \cite{ginsburg67} there exists a non-deterministic G.G.H. machine $M$ which accepts the intersection of the sets of inputs accepted by the all machines $M_{i}$ and $N_{i}$ for $i=1,\ldots, n$. Now the decision
theorems of G.G.H. do not apply to non-deterministic G.G.H. PDAs. However, the method of La Torre~\cite{latorre07} does apply. The machine M can be represented in Rabin's theory WS2S~\cite{rabin69} and
tested for emptiness.
\end{proof}

Many CQ decision problems can be solved with this method. For example, it is decidable
whether the set of CQ distinct products of members of $B$ is infinite. This will be seen in the next
section.

\section{MORE APPLICATIONS OF G.G.H. PDA'S}
There is a well-known duality between rooted binary trees with $n$ leaves and triangulations of an
$n$-gon. In the CQ case the leaves come equipped with shifts. These shifts control the results of
further composition. The nice geometrical questions about tiling the plane with $n$-gons are
analogous to questions about covering Cantor Space by normal CQ expressions.

\begin{definition}
Let $B$ be a finite set of CQ normal forms. We say that $B$ \textit{covers Cantor space} if there
exists an infinite sequence
$F_1, \ldots, F_ {n}, \ldots$ of members of $B$ s.t. for each shift
$S$ there exists an $n$ s.t. the normal form of
$S*F_ {1} * \ldots * F_{n}$
is not a shift.
\end{definition}

\noindent
As an example, the set $\{ \langle I, I \rangle \}$ covers Cantor space but $\{ \langle R, L \rangle \}$ does not.

\begin{definition}
A shift $S$ is said to be \textit{bad} for B if for any sequence
$F_ {1}, \ldots, F_{n}$ of members of $B$ the CQ normal form of
$S*F_ {1} *...*F_{n}$
is a shift.
\end{definition}

\noindent
For example, $R*R$ is bad for $\{ \langle L, R*R \rangle \}$. First we observe that the set of bad B shifts is recursive uniformly in B. This can be seen by constructing a deterministic G.G.H. one-way stack machine which, if started with $S$ in its stack, accepts all inputs if and only if $S$ is bad. The
construction of the machine is uniform in $S$, so that the set of $S$ which are bad is definable in Rabin's theory WS2S by LaTorre's method. To recall, the input alphabet is the set $B$ and the
machine runs as follows. With input symbol $F$ the machine reads the top of the current stack
looking for a minimal $S$ s.t. the normal form of $S*F$ is a shift. This can be done by a DFSA.
Having found such an $S$, the machine pops $S$ and pushes the normal form of $S*F$. Otherwise, the
machine rejects the entire input.

Before considering coverings, we settle an algebraic question.
\begin{theorem}
It is decidable whether the submonoid generated by $B$ is infinite.
\end{theorem}
\begin{proof}
Of course, if $B$ contains a shift then the submonoid generated by $B$ is infinite, so we can
assume this is not the case. First, we say that an $S$ bad for $B$ is \textit{extenuative} if for each $n$ there
exists
$F_ {1}, \ldots, F_{k}$ in $B$ s.t. the CQ normal form of $S*F_ {1} *...*F_{k}$
has length at least $n$. Now it is decidable whether $S$ is extenuative.
First decide whether $S$ is bad for $B$. Now consider the set of natural numbers encoded as strings
of $L$s and the input language $B + \{L\}$.

Design a deterministic G.G.H. PDA with $S$ initially in its stack and which accepts an input $W$ if
and only if $W$ has the form
$F_ {1} \ldots F_{k} L \ldots L$ (with $n$ occurrences of $L$), and where after reading $F_ {1} \ldots F_{k}$ the stack has length at least $n$.

Now the set of words $L... L$ s.t. $F_ {1}, \ldots, F_{k}L... L$ is accepted by the machine is the result of applying a sequential transducer to the set of words $F_ {1}, \ldots, F_{k}L \ldots L$ accepted. Thus by \cite[Theorem~2.4]{ginsburg67} we can construct a non-deterministic G.G.H PDA which accepts exactly this set of words $L... L$. Thus by La~Torre's method we can decide if this set is all strings of $L$s.

Next pick an initial member $F$ of $B$ and one of its shifts $S$. We construct a G.G.H. nondeterministic PDA as
follows. The input alphabet consists of $B$ plus a new letter $@$ (we could have used $L$ as above).
On inputs of the form $F_ {1} ... F_{k}@_{m}$ the machine proceeds as above except when the
current stack is $@_p S'$ and the input is $F_{i}$ s.t. the CQ
normal form of $S' * F_{i}$ is not a shift. In this case, the machine guesses a minimal length shift
$S''$ s.t. the normal form $S'' * S' * F_{i}$
is a shift, say $S+$, and updates the stack to $@_{p+1} S+$. The machine accepts if in the final stack
the number of $@$ is not less than $m$.

By G.G.H. Theorem 2.4 the set of $@_{m}$ s.t. there is a $F_ {1}, \ldots, F_{k}@_m$
accepted by the machine is accepted by a non-deterministic G.G.H PDA. Now it is decidable if
this set is all $@_{m}$. We distinguish two cases.

\textit{Case 1.} The set is all $@_{m}$. Then there are arbitrarily large CQ normal form generated by $B$ and submonoid generated by $B$ is infinite.

\textit{Case 2.} The set of all $@_{m}$ is finite for all $F$ in $B$ and $S$. Then every infinite sequence $F_ {1}, \ldots, F_ {n}, \ldots$ of members of $B$ must have an initial segments $F_ {1}, \ldots, F_{k}$ s.t. all the shifts of the CQ normal form of $F_ {1} * \ldots  *F_{k}$ are bad. So by Konig's lemma there exists a finite tree of
such finite sequences s.t.
every path has an initial segment in this tree. Now search until such a finite tree $T$ is found. Now
suppose that the CQ normal form of $F_ {1} *...*F_{k}$ has a bad shift S which is extenuative.
Now $F_ {1}, \ldots, F_{k}$ has an initial segment in $T$ say $F_ {1}, \ldots, F_{p}$.
So $F_ {1} * \ldots *F_{p}$ has a bad shift with the same property. Now the submonoid generated by B
will be finite if and only if no such extenuative shift exists. Thus it suffices to test all the shifts of
products in the tree for extenuativeness.
\end{proof}

\begin{theorem}
It is decidable whether $B$ covers Cantor space.
\end{theorem}

\begin{proof}
Given $B$ it is impossible to
cover Cantor space if there is a bad shift i.e. a shift s.t. for all $F_ {1} \ldots F_{n}: B$ the normal
form of $S*F_ {1} * \ldots *F_{n}$ is a shift. This is decidable by previous argument. Now if no bad
shift exists then for every shift $S$ there exits $F_ {1} ...F_{n}: B$ s.t. the normal form of $S*F_ {1}
* \ldots *F_{n}$ is not a shift. Thus $B$ can cover Cantor space
by repeatedly applying the $F_ {1} *...*F_{n}$ as shifts $S$ appear in the normal form of previous
applications.
\end{proof}

We conclude with an amusing observation.

\section{GIGSAW PUZZLES}
     A gigsaw puzzle is a patern matching problem where each variable
occurs at most once and solutions come from a fixed set of CQ expressions
all of which must be used. Here we show the problem is NP complete.

    We encode the satisfiability problem. We assume
that we are given a conjunctive normal form (conjunction of disjunctions;
we regard $x \vee x \vee y$ as distinct from $x \vee y$). We suppose that the variables
are $x_{1}, \ldots, x_{n}$. For each variable $x_{i}$ we construct two gadgets
$G_{i}$ and $H_{i}$ by
\begin{align*}
G_{i} &= \langle \underbrace{\langle I,I \rangle, \langle \langle I,I \rangle , \langle ...} \langle L, \langle \langle \underbrace{I,I \rangle , \langle ... \langle \langle I,I \rangle} ,I \rangle ...\rangle \rangle \rangle ... \rangle \rangle \rangle
\\*[-4pt]
&\hspace{14pt}
\text{$i-1$ occurrences}
\hspace{15pt}
\text{$n-i+1$ occurrences}
\\[5pt]
H_{i} &= \langle \underbrace{\langle I,I \rangle , \langle \langle I,I \rangle, \langle ...} \langle R, \langle \langle \underbrace{I,I \rangle , \langle ... \langle \langle I,I \rangle} , I \rangle ... \rangle \rangle \rangle ... \rangle \rangle \rangle
\\*[-4pt]
&\hspace{14pt}
\text{$i-1$ occurrences}
\hspace{15pt}
\text{$n-i+1$ occurrences}
\end{align*}

Now suppose that we have a conjunct C of the form
$$x_{a(1)} \bigvee \ldots \bigvee x_{a(k)} \bigvee \neg x_{b(1)} \bigvee ... \bigvee \neg x_{b(m)}$$
We replace each of the $k+m$ variable occurrences by a new variable $y_{i}$
and we construct a product $C\#$ of terms
\begin{align*}
&L*R^{a(i)} * y_{i} * \langle I,  \langle I, I \rangle \rangle&&\text{ if $i< k+1$}
\\
&L*R^{b(k-i)} * y_{i} * \langle \langle I, I \rangle, I\rangle&&\text{ if $k < i$}
\end{align*}
and the identity $C\# = I$. Now assume that $x_{i}$ occurs $m$ times. We
introduce $m$ new variables $z_{1}, \ldots, z_{m}$ and define the term $B$
in $m$ stages as follows
\begin{align*}
B_ {1} &= L*z_ {1}
\\
B_{j+1} &= B_{j}*z_{j}* \langle \langle L, \langle I, I \rangle \rangle, \langle \langle I, I \rangle, R \rangle \rangle
\\
B &= B_{m}
\end{align*}
and we add the identity $B * \langle I,I \rangle = I$. This is a total number of identities
equal to the number of conjuncts plus the number of variables. Then
the identities are solvable using the gadgets if and only if the original conjunction
is satisfiable.

\section{WHAT NEXT?}
It is interesting to see if these methods can be extended to the free Cartesian monoid. The
questions one wants to ask about the corresponding non-deterministic G.G.H. machines do not
seem to be answerable in any straightforward manner. However, we can make some direct
applications. Let RI be the submonoid of right invertible elements of CM. Here, we let $B$ be a
finite subset of RI and $F$ an element of RI, all in (1)--(7) normal form.

\begin{lemma}
If $F$ is in the submonoid generated by $B$ then there exists
$F_ {1}, \ldots, F_{n}$ in B s.t.
$$F_ {1} * \cdots * F_{n} \hooktto_{1,2,3,4,5} F$$
except for $F = I$.
\end{lemma}

\begin{proof}
First recall the characterization of the elements in RI given in~\cite{statman96}. $F$ is in RI if and only if no shift
of $F$ is a final segment (left to right) of another shift of $F$. Now suppose that
\mbox{$F_ {1} *...*F_{n} \hooktto_{1,2,3,4,5,6,7} F$}.
By Theorem 1 (3) there exists $G$ s.t.
$F_ {1} *...*F_{n} \hooktto_{1,2,3,4,5} G \hooktto_{5,6,7} F$. Note that in the rewrite from $G$ to $F$ all the applications of (6) are to cancel a shift. Now let $J$ be the
kernel of the unique homomorphism of CQ on CM. $J$ is precisely the set of all $H$ in CQ s.t. $H = I$
in CM. Now consider the shape of $G$. The binary tree of $G$ begins with the tree of $F$ but where a
shift $S$ of $F$ would occur is the CQ normal form of a member of $J * S$. Note that such an $S$ is not $I$
since $F$ is in RI, but there is the trivial case when the member of $J$ is $I$.

Now each member $H$ of $J$
has the following structure. Each leaf of the binary tree of $H$ can be described by a sequence of
$L$s and $R$s read from right to left. The corresponding shift $S$ has the property that the normal
form of $S* H$ is precisely $S$, and no final segment (left to right) $S''$ of $S'$ has the property that
$S''*H$ equals a shift in CQ. Now let $H$ be the member of $J$ for $S$ a shift of $F$. We have for each shift $S'$
of $H$ s.t. if $S$ is in position $S''$, then $S' * S'' * G = S' * S$ in CQ.
Now, for each $i=1,..., n$, $S'' * S' * F_ {1} *...*F_{i}$ equals a shift in CQ.
However, since all the $F_{j}$ are in RI there is a unique shortest $S' * S$ with this property s.t.
$S' * S'' * G = S' * S$ in CQ. But if $H$ is not trivial then there are at least two such. Thus $H$ is trivial.
Hence $G = F$.
\end{proof}

\begin{corollary*}
Theorems 3, 4 and 5 apply to RI, and thus to the Thompson-Higman groups.
\end{corollary*}

\nocite{*}
\bibliographystyle{eptcs}
\bibliography{generic}
\end{document}